\documentclass[journal]{IEEEtran}
\usepackage[utf8]{inputenc}
\usepackage{graphicx}
\usepackage{cite}
\usepackage{balance}
\usepackage{amsmath,amsthm,amssymb}
\usepackage{epstopdf}
\usepackage{threeparttable}
\usepackage{multirow}
\usepackage[colorlinks, citecolor=blue]{hyperref}
\usepackage{color}
\usepackage{caption}
\usepackage{subfig}
\usepackage{mathtools}
\usepackage{hyperref}
\usepackage{lscape}
\newcounter{MYtempeqncnt}
\usepackage{float}
\usepackage{diagbox}
\usepackage{stfloats}
\usepackage{booktabs} 
\newtheorem{theorem}{Theorem}

\newtheorem{proposition}[theorem]{Proposition}

\ifCLASSINFOpdf
\else
\fi

\begin{document}
\title{On the Impacts of Phase Shifting Design and Eavesdropping Uncertainty on Secrecy Metrics of RIS-aided Systems}
\author{\IEEEauthorblockN{Long Kong{$^\ddagger$},~Steven Kisseleff{$^\dagger$},~Symeon Chatzinotas{$^\dagger$},~Björn Ottersten{$^\dagger$}, and Melike Erol-Kantarci{$^\ddagger$}}\\
\IEEEauthorblockA{$^\ddagger$ School of Electrical Engineering and Computer Science, University of Ottawa, Ontario, Canada\\
$^{\dagger}$The Interdisciplinary Centre for Security Reliability and Trust (SnT), University of Luxembourg, Luxembourg  \\
Emails: lkong2@uottawa.ca, \{steven.kisseleff, symeon.chatzinotas, bjorn.ottersten\}@uni.lu, melike.erolkantarci@uottawa.ca}}

\maketitle
\begin{abstract}

This paper investigates the secrecy outage probability (SOP), the lower bound of SOP, and the probability of non-zero secrecy capacity (PNZ) of reconfigurable intelligent surface (RIS)-assisted systems from an information-theoretic perspective. In particular, we consider the impacts of eavesdroppers' location uncertainty and the phase adjustment uncertainty, namely imperfect coherent phase shifting and discrete phase shifting on RIS. 
More specifically, analytical and simulation results are presented to show that (i) the SOP gain due to the increase of the RIS reflecting elements number gradually decreases; and (ii) both phase shifting designs demonstrate the same PNZ secrecy performance, in other words, the random discrete phase shifting outperforms the imperfect coherent phase shifting design with reduced complexity. 
\end{abstract}
\IEEEpeerreviewmaketitle
\section{Introduction}
Reconfigurable intelligent surface (RISs, a.k.a. large intelligent metasurfaces and passive holographic multiple-input and multiple-output (MIMO) surfaces) is envisioned as an appealing candidate to boost the energy and spectral efficiency for the fifth-generation (5G) of communication systems and beyond \cite{9140329,9424177}. As stated in \cite{9140329,9424177}, RIS has the potential to create smart radio environments and can be easily integrated with many existing communication techniques, e.g., multiple-input multiple-output (MIMO) \cite{9201173}, device-to-device (D2D) \cite{9305710}, and non-orthogonal multiple access (NOMA) \cite{9385957}, etc. Different from the conventional relaying technique \cite{9095301}, RIS can be properly configured to alter its incident signals' reflection directivity to either enhance signals' quality or to create destructive interference towards unauthorized users. The latter potential application is technically feasible to enhance security from the information theoretic perspective, namely physical layer security (PLS) \cite{9134962,9201173,9305710,9328149,9385957,9439833,9453160,9350290}. 

Considering the classic Wyner's wiretap model, numerous research efforts were put into investigating the secrecy performance of the RIS-aided wireless communication systems. For instance, Ai \textit{et al}. in \cite{9453160} derived the secrecy outage probability (SOP) for the RIS-aided vehicular networks. Trigui \textit{et al}. in \cite{9350290} considered the colluding and non-colluding passive eavesdropping scenarios of the RIS-aided communications using quantized phase shifting design, and derived the SOP performance metric. Given the spatial distributions of multiple users, Zhang \textit{et al}. in \cite{9439833} exploited the stochastic geometry to model the randomness of users' location and also derived the SOP and the probability of non-zero secrecy capacity. Inspired by an interesting insight revealed by Karas in \cite{7543509} suggesting that both the influences of pathloss and eavesdropper's location uncertainty should be taken into account in the deployment of wireless systems with secrecy constraint, in this paper we henceforth investigate the impacts of RIS phase shifting design and eavesdropper's location uncertainty (due to its passive nature) on the secrecy performances. 

To this end, the contributions of this paper are three-fold:
\begin{enumerate}
\item Considering the impacts of imperfect coherent phase shifting and random discrete phase shifting designs and using the stochastic geometry to model eavesdropper's location uncertainty, the statistics of legitimate and illegitimate users are derived, respectively. 
\item Exploiting the obtained statistics, i.e., the probability density function (PDF) and cumulative distribution function (CDF), secrecy performance metrics, including SOP, lower bound of SOP, and probability of non-zero secrecy capacity (PNZ), are derived with closed-form expressions.
\item Confirming the accuracy of analytical results with Monte-carlo (MC) simulations. Numerical results display that (i) merely increasing the number of reflecting elements will not enhance the secrecy outage probability significantly; and (ii) the random discrete phase shifting design outperforms the imperfect phase shifting design with reduced system overhead and without presenting secrecy performance loss. 
\end{enumerate}
\textit{Mathematical Functions and Notations}: $[x]^+ = \max(0,x)$, $j \triangleq \sqrt{-1}$, $\Gamma(\cdot)$, $\mathbf{\gamma}(a,b)$, and $\mathbf{\Gamma}(a,b)$ are the complete, lower incomplete, and upper incomplete Gamma functions, $[x]^+=\max(x,0)$; $H_{p,q}^{m,n}[\cdot]$ is the univariate Fox's $H$-function \cite[Eq. (1.2)]{mathai2009h}; $\mathbf{K}_n(\cdot)$ is the modified Bessel function of the second kind \cite[Eq. (5.52)]{gradshteyn2014table};
$G_{p,q}^{m,n}[\cdot]$ is the univariate Meijer's $G$-function \cite[Eq. (9.301)]{gradshteyn2014table}.
\section{System model and Phase shifting design}
\subsection{System Model}
Consider a typical RIS-aided wiretap model, where a transmitter (Alice) intends to send confidential messages to a legitimate receiver (Bob) with the assistance of an $L$-element RIS in the presence of a passive eavesdropper (Eve). It is assumed that (i) Alice, Bob, and Eve are equipped with single antennas; (ii) all channel gains are complex Gaussian distributed with zero means and unit variances; (iii) Eve is randomly located within a ring centred at the RIS’ position according to homogeneous Poisson point processes (HPPP), where the ring outer radius $R_2$ indicates its coverage region, and the inner radius $R_1$ depicts the minimum distance between an RIS and Eve required to achieve a predetermined SOP, for the design and deployment of a system with specific SOP requirements; and (iv) the distance $d_{AR}$ between Alice and RIS is constant.

The instantaneous received signal-to-noise ratio (SNR) for Bob and Eve, $\gamma_i, i\in \{B,E \}$ are respectively given by \cite{Kong2022WCNC}
\begin{align}
\gamma_i = \frac{P}{ \mathcal{N}_0 d_{AR}^\delta d_{R,i}^\delta} \left \vert \underbrace{\sum \limits_{l=1}^L \mathfrak{h}_{l,i} \mathfrak{g}_{l,i} \exp(j\epsilon_{l,i})}_{\epsilon_L} \right\vert^2,
\end{align}
where $P$ is the transmit power, $h_{l,i}=\mathfrak{h}_{l,i} \exp(-j\phi_{l,i}) $ and $g_{l,i} =\mathfrak{g}_{l,i} \exp(-j\psi_{l,i}) $ are independent and identically distributed (i.i.d.) Rayleigh RVs, which correspond to the channel coefficients from Alice to the $l$-th reflector element and the $l$-th reflector element to Bob or Eve, i.e., $i \in \{B,E\}$, respectively. $\mathfrak{h}_{l,i}$, $\mathfrak{g}_{l,i}$, $\phi_{l,i}$, and $\psi_{l,i}$ denote the amplitudes and phases of the corresponding fading channel gains. $u_{l,i} = w_{l,i}(\theta_{l,i})\exp(j\theta_{l,i})$ is the reflection coefficient produced by the $l$-th element of the RIS, herein $w_{l,i}(\theta_{l,i})=1$ for the phase shifts, $l=1,\cdots,L$, and $\epsilon_{l,i} = \theta_{l,i} -\phi_{l,i}-\psi_{l,i} $. $\delta$ is the path loss exponent. $\mathcal{N}_0$ is the variance of the the additive white Gaussian noise (AWGN). For the simplicity of notations, let $\bar{\gamma}_i=\frac{P}{\mathcal{N}_0d_{AR}^\delta}$ is the average SNR.
\subsection{Imperfect Coherent Phase Shifting}
Assume the phases of $h_l$ and $g_l$ can be acquired from Alice, the imperfect coherent phase shifting design means that the phase shifting of the RIS is imperfectly matched with the RIS fading gains, which is denoted by $\epsilon$ following the uniform distribution, i.e., $\epsilon_{l,i} \sim U(-\pi,\pi)$. Accordingly, the CDF and PDF of $\gamma_B$ are given by \cite[Eq. (43)]{9138463}\cite[Eq. (5)]{Kong2022WCNC}
\begin{align} \label{CDF_gammaB}
F_B^C(\gamma)& = 1- \frac{2^{1-L}}{\Gamma(L) }\left( \sqrt{\frac{\gamma}{\bar{\gamma}_B}}\right)^L \mathbf{K}_L\left(  \sqrt{\frac{\gamma}{\bar{\gamma}_B}} \right) \nonumber \\
& \mathop=^{(a)} 1- \frac{1}{\Gamma(L) }G_{0,2}^{2,0} \left[ { \frac{\gamma }{4\bar{\gamma}_B}\left|  {\begin{array}{*{20}c}
    {-}   \\
   { L,0}  \\
\end{array}} \right.}  \right] \nonumber \\
& = 1 - \bar{F}_B^C(\gamma),
\end{align}
\begin{align} \label{PDF_gammaB}
f_B^C(\gamma)& = \frac{1}{2^L\bar{\gamma}_B   \Gamma(L)}\left( \sqrt{\frac{\gamma}{\bar{\gamma}_B}} \right)^{L-1} \mathbf{K}_{L-1}\left(  \sqrt{\frac{\gamma}{\bar{\gamma}_B}} \right)\nonumber \\
& = \frac{1}{\Gamma(L) }G_{0,2}^{2,0} \left[ { \frac{\gamma }{4\bar{\gamma}_B}\left|  {\begin{array}{*{20}c}
    {-}   \\
   { L-1,0}  \\
\end{array}} \right.}  \right],
\end{align}
where step $(a)$ is developed using \cite[Eq. (9.34.3)]{gradshteyn2014table}, and $\bar{F}_{B}^C(\gamma)$ is the complementary CDF of $F_B^C(\gamma)$.
\begin{proposition} 
For imperfect coherent phase shifting design, the PDF and CDF of $\gamma_E$ are respectively given by (\ref{PDF_gammaE}) and (\ref{CDF_gammaE}), shown at the top of next page.
\begin{figure*}[!t]
\setcounter{MYtempeqncnt}{\value{equation}}
\setcounter{equation}{3}
\begin{align} \label{PDF_gammaE}
f_E^C(\gamma)=& \frac{1}{2\delta \bar{\gamma}_E \Gamma(L)(R_2^2-R_1^2)}\left[R_2^{2+\delta}G_{1,3}^{2,1} \left[ {\frac{R_2^\delta \gamma}{4\bar{\gamma}_E}   \left| {\begin{array}{*{20}c}
		{-\frac{2}{\delta}}   \\
		{L-1,0,-1-\frac{2}{\delta} }  \\
		\end{array}} \right.}  \right] - R_1^{2+\delta}G_{1,3}^{2,1} \left[ {\frac{R_1^\delta\gamma}{4\bar{\gamma}_E}   \left| {\begin{array}{*{20}c}
		{-\frac{2}{\delta}}   \\
		{L-1,0,-1-\frac{2}{\delta} }  \\
		\end{array}} \right.} \right]\right],
\end{align}
\begin{align} \label{CDF_gammaE}
F_E^C(\gamma)& =  \frac{2}{\delta \Gamma(L)(R_2^2-R_1^2)}\left[R_2^{2}G_{2,4}^{2,2} \left[ {\frac{R_2^\delta\gamma}{4\bar{\gamma}_E}   \left| {\begin{array}{*{20}c}
		{1,1-\frac{2}{\delta}}   \\
		{L,1,-\frac{2}{\delta},0 }  \\
		\end{array}} \right.} \hspace{-1.2ex} \right] - R_1^{2}G_{2,4}^{2,2} \left[ {\frac{R_1^\delta \gamma}{4\bar{\gamma}_E}   \left| {\begin{array}{*{20}c}
		{1,1-\frac{2}{\delta}}   \\
		{L,1,-\frac{2}{\delta},0  }  \\
		\end{array}} \right.} \right]\right],
\end{align}
\hrulefill
\end{figure*}
\end{proposition}
\begin{proof}
The PDF of $\gamma_E$ refers to \cite[Eq. (3)]{Kong2022WCNC}, and then integral in terms of $\gamma$, the proof for (\ref{CDF_gammaE}) is completed.
\end{proof}
\begin{figure}[!t]
\centering{\includegraphics[width=\columnwidth]{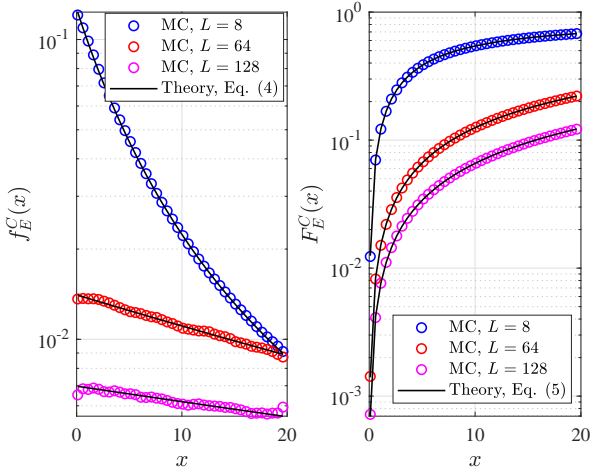}}
\caption{PDF \& CDF of $\gamma_E$, $\delta = 3.7, R_1 = 0.1, R_2 = 1$ and $\bar{\gamma}_E = 40$ dB.}
\label{PDF_CDF_gammaE_imperfect}
\end{figure}
\subsection{Random Discrete Phase Shifting}
Exploiting random discrete phase shifting can avoid the phase adjustment with perfect phase resolution and  reduce system overhead needed for global channel state information (CSI) acquisition for Alice. Hence, the random discrete phase shifting design is assumed. More specifically, the phase shifts are selected randomly from the following sets
\begin{equation}
\epsilon_{l,i} \in \left\{0,\frac{2\pi}{2^Q},\cdots,\frac{2\pi (2^Q -1)}{2^Q} \right\},
\end{equation}
where $Q$ is the resolution of the RIS element’s phase shifter \cite{8741198}.

Leveraging the results \cite[Lemma 2]{9079918}, when $L \rightarrow \infty$, $\epsilon_L$ can be approximated as a complex Gaussian random variable (RV) with zero mean and $L$ variance, i.e., $\epsilon_L \sim \mathcal{CN}(0,L)$. Using the fact that $|\epsilon_L|^2$ follows the exponential  distribution with parameter $L$. The PDF and CDF of $\gamma_B$ are given, respectively, by
\begin{equation} \label{PDF_gammaB_Discrete}
f_B^D(\gamma) = \frac{1  }{L\bar{\gamma}_B} \exp\left(-\frac{\gamma}{L\bar{\gamma}_B}\right), 
\end{equation}
\begin{equation} \label{CDF_gammaB_Discrete}
F_B^D(\gamma) = 1 - \exp\left(-\frac{\gamma}{L\bar{\gamma}_B}\right).
\end{equation}
\begin{proposition} \label{Proposition_GammaE_Discrete}
For random discrete phase shifting design, when $L$ tends to large values, the PDF and CDF of $\gamma_E$ can be tightly approximated by   
\begin{align} \label{PDF_gammaE_Discrete}
f_E^D(\gamma) & = \frac{2 }{\delta L \bar{\gamma}_E (R_2^2 - R_1^2)} \left( \frac{L\bar{\gamma}_E}{\gamma}\right)^{1+\frac{2}{\delta}} \nonumber \\
&\times \left[\gamma\left(1+\frac{2}{\delta},\frac{R_2^\delta}{L\bar{\gamma}_E}\gamma\right)-\gamma\left(1+\frac{2}{\delta},\frac{R_1^\delta}{L\bar{\gamma}_E}\gamma\right)\right],
\end{align}
\begin{align}  \label{CDF_gammaE_Discrete}
F_E^D(\gamma) =& 1- \frac{2}{\delta(R_2^2-R_1^2)}\left(\frac{L\bar{\gamma}_E}{\gamma}\right)^{\frac{2}{\delta}} \nonumber \\
& \hspace{5ex}\times \left[\gamma\left(\frac{2}{\delta},\frac{R_2^\delta \gamma}{L\bar{\gamma}_E} \right)-\gamma\left(\frac{2}{\delta},\frac{R_1^\delta \gamma}{L\bar{\gamma}_E} \right) \right].
\end{align}
\end{proposition}
\begin{proof}
See Appendix. \ref{Proof_Discrete_SNR_E}.
\end{proof}
\begin{figure}[!t]
\centering{\includegraphics[width=\columnwidth]{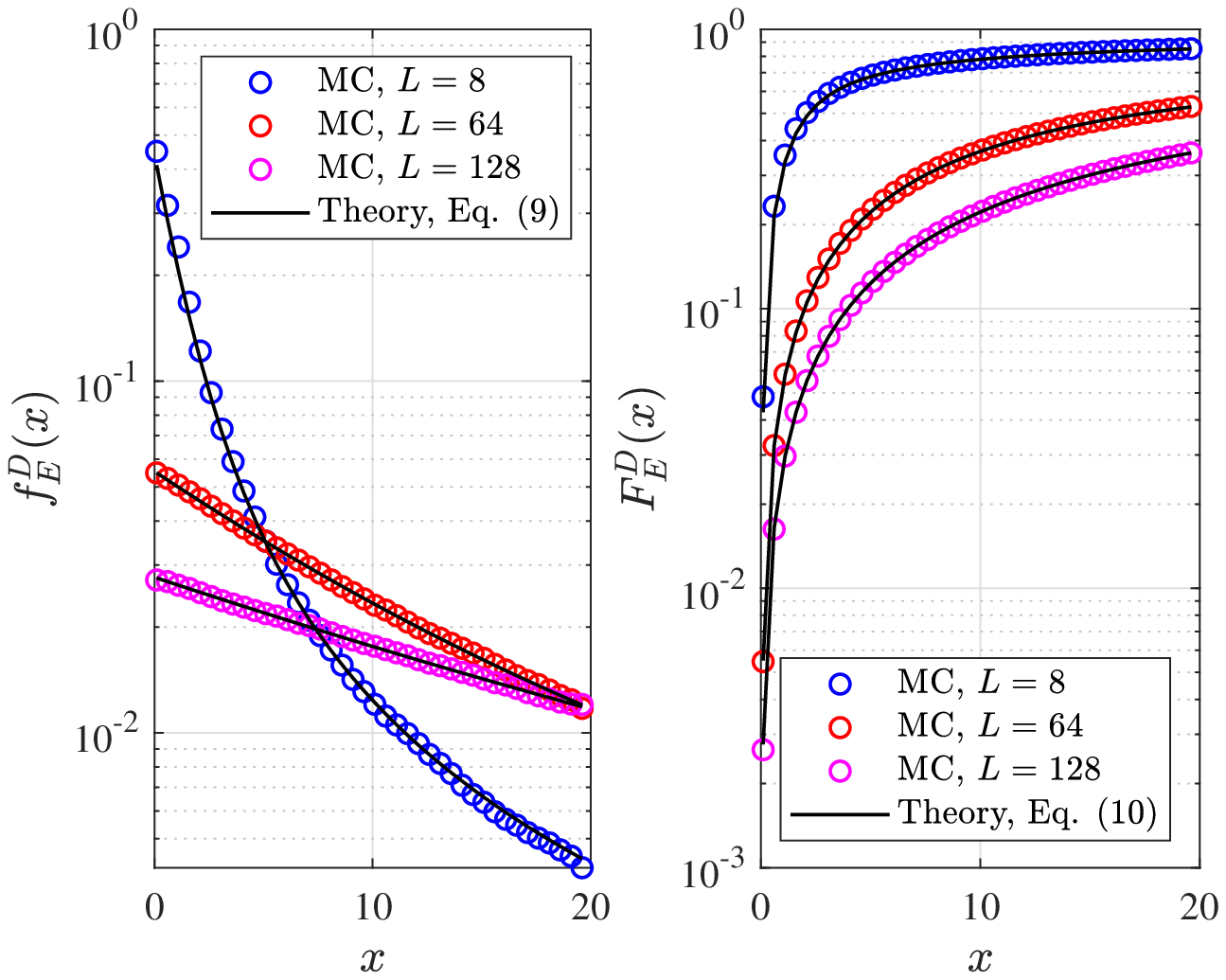}}
\caption{PDF \& CDF of $\gamma_E$, $\delta = 3.7, R_1 = 0.1, R_2 = 1$ and $\bar{\gamma}_E = -10$ dB.}
\label{PDF_CDF_gammaE_Discrete}
\end{figure}
Figs. \ref{PDF_CDF_gammaE_imperfect} and \ref{PDF_CDF_gammaE_Discrete} present the analytical PDFs and CDFs with the simulated ones, apparently the tightness of our analytical results are verified. 
\section{Secrecy Performance characterization}
According to \cite{8706707}, the instantaneous secrecy capacity of such a system configuration under the assumption that eavesdroppers do not collude is given by
\begin{equation}
C_{s} = [\log_2(1+\gamma_B) - \log_2(1 + \gamma_E)]^+.
\end{equation}

Based on $C_s$, the SOP, the lower bound of SOP, and the PNZ are correspondingly derived in this section.
\subsection{Secrecy characterization with imperfect coherent phase shifting design}
\begin{proposition}
Considering the imperfect coherent phase shifting design, the SOP ($\mathcal{P}_{out}^C$), lower bound of SOP ($\mathcal{P}_{out}^{C,L}$), and PNZ ($\mathcal{P}_{nz}^C$) of the RIS-assisted system under consideration are respectively given by (\ref{SOP_exact}), (\ref{LowSOP_exact}), and (\ref{PNZ_exact}), shown at the top of next page.
\begin{figure*}[!t]
\setcounter{MYtempeqncnt}{\value{equation}}
\setcounter{equation}{11}
\begin{align} \label{SOP_exact}
\mathcal{P}_{out}^C = 1 - \frac{2\bar{\gamma}_B}{R_s\delta \bar{\gamma}_E \Gamma(L)\Gamma(L)(R_2^2-R_1^2)}\left[R_2^{2+\delta} \sum \limits_{k=0}^\infty \frac{1}{k!}\left( \frac{-W}{4\bar{\gamma}_B }\right)^kG_{3,3}^{2,3} \left[ {\frac{R_2^\delta \bar{\gamma}_B}{\bar{\gamma}_E R_s}   \left| {\begin{array}{*{20}c}
		{0,-\frac{2}{\delta},k-L}   \\
		{L-1,0,-1-\frac{2}{\delta} }  \\
		\end{array}} \right.} \right] \right. \nonumber \\
 \left.- R_1^{2+\delta} \sum \limits_{k=0}^\infty \frac{1}{k!}\left( \frac{-W}{4\bar{\gamma}_B}\right)^kG_{3,3}^{2,3} \left[ {\frac{R_1^\delta\bar{\gamma}_B}{\bar{\gamma}_ER_s}   \left| {\begin{array}{*{20}c}
		{0,-\frac{2}{\delta},k-L}   \\
		{L-1,0,-1-\frac{2}{\delta} }  \\
		\end{array}} \right.} \right] \right] ,
\end{align}
\begin{align} \label{LowSOP_exact}
\mathcal{P}_{out}^{C,L} = 1 - \frac{2}{\delta \bar{\gamma}_E \Gamma(L)\Gamma(L)(R_2^2-R_1^2)}\left[R_2^{2}G_{3,3}^{2,3} \left[ {\frac{R_2^\delta \bar{\gamma}_B}{R_s\rho_E}   \left| {\begin{array}{*{20}c}
		{1-\frac{2}{\delta},1-L,1}   \\
		{L,1,-\frac{2}{\delta} }  \\
		\end{array}} \right.} \hspace{-1.2ex} \right] 
		 - R_1^{2}G_{3,3}^{2,3} \left[ {\frac{R_1^\delta\bar{\gamma}_B}{R_s\bar{\gamma}_E}   \left| {\begin{array}{*{20}c}
		{1-\frac{2}{\delta},1-L,1}   \\
		{L,1,-\frac{2}{\delta} }  \\
		\end{array}} \right.} \right]\right],
\end{align}
\begin{align} \label{PNZ_exact}
\mathcal{P}_{nz}^C=\frac{2}{\delta \Gamma(L)\Gamma(L)(R_2^2-R_1^2)}\left( R_2^2 G_{3,3}^{2,3} \left[ {\frac{R_2^\delta \bar{\gamma}_B}{\bar{\gamma}_E}   \left| {\begin{array}{*{20}c}
		{1,1-\frac{2}{\delta},1-L}   \\
		{L,1,-\frac{2}{\delta} }  \\
		\end{array}} \right.} \hspace{-1.2ex} \right]  -R_1^2G_{3,3}^{2,3} \left[ {\frac{R_1^\delta \bar{\gamma}_B}{\bar{\gamma}_E}   \left| {\begin{array}{*{20}c}
		{1,1-\frac{2}{\delta},1-L}   \\
		{L,1,-\frac{2}{\delta} }  \\
		\end{array}} \right.} \hspace{-1.2ex} \right]  \right).
\end{align}
\hrulefill
\end{figure*}
\end{proposition}
\begin{proof}
See Appendix \ref{Proof_SecrecyMetric_uncertainty}.
\end{proof}
\subsection{Secrecy characterization with random discrete phase shifting design}
\begin{proposition}
Considering the random discrete phase shifting design, the SOP ($\mathcal{P}_{out}^D$), lower bound of SOP ($\mathcal{P}_{out}^{D,L}$), and PNZ ($\mathcal{P}_{nz}^D$) of the RIS-assisted system under consideration are respectively given by (\ref{SOP_exact_Discrete}), (\ref{LowSOP_exact_Discrete}), and (\ref{PNZ_exact_Discrete}), shown at the top of next page.
\begin{figure*}[!t]
\setcounter{MYtempeqncnt}{\value{equation}}
\setcounter{equation}{14}
\begin{equation} \label{SOP_exact_Discrete}
\mathcal{P}_{out}^D = 1 - \frac{2}{\delta(R_2^2 - R_1^2)}\left(\frac{R_s\bar{\gamma}_E}{\bar{\gamma}_B} \right)^\frac{2}{\delta} \sum \limits_{k=0}^\infty \frac{1}{k!} \left(-\frac{W}{L\bar{\gamma}_B} \right)^k \left[ G_{2,2}^{1,2} \left[ { \frac{R_2^\delta \bar{\gamma}_B }{R_s \bar{\gamma}_E}\left|  {\begin{array}{*{20}c}
    {1+\frac{2}{\delta},1}   \\
   {1+\frac{2}{\delta},0}  \\
\end{array}} \right.}  \right] - G_{2,2}^{1,2} \left[ { \frac{R_1^\delta \bar{\gamma}_B }{R_s \bar{\gamma}_E}\left|  {\begin{array}{*{20}c}
    {1+\frac{2}{\delta},1}   \\
   {1+\frac{2}{\delta},0}  \\
\end{array}} \right.}  \right] \right] ,
\end{equation}
\begin{align} \label{LowSOP_exact_Discrete}
\mathcal{P}_{out}^{D,L} = 1 - \frac{2}{\delta(R_2^2 - R_1^2)}\left(\frac{R_s\bar{\gamma}_E}{\bar{\gamma}_B} \right)^\frac{2}{\delta} \left[ G_{2,2}^{1,2} \left[ { \frac{R_2^\delta \bar{\gamma}_B }{R_s \bar{\gamma}_E}\left|  {\begin{array}{*{20}c}
    {1+\frac{2}{\delta},1}   \\
   {1+ \frac{2}{\delta},0}  \\
\end{array}} \right.}  \right]- G_{2,2}^{1,2} \left[ { \frac{R_1^\delta \bar{\gamma}_B }{R_s \bar{\gamma}_E}\left|  {\begin{array}{*{20}c}
    {1+\frac{2}{\delta},1}   \\
   {1+ \frac{2}{\delta},0}  \\
\end{array}} \right.}  \right]\right],
\end{align}
\begin{align} \label{PNZ_exact_Discrete}
\mathcal{P}_{nz}^{D} = 1 - \frac{2}{\delta(R_2^2-R_1^2)}\left(\frac{\bar{\gamma}_E}{ \bar{\gamma}_B} \right)^{\frac{2}{\delta}}\left[ G_{2,2}^{1,2} \left[ { \frac{R_2^\delta \bar{\gamma}_B }{\bar{\gamma}_E}\left|  {\begin{array}{*{20}c}
    {\frac{2}{\delta},1}   \\
   { \frac{2}{\delta},0}  \\
\end{array}} \right.}  \right] - G_{2,2}^{1,2} \left[ { \frac{R_1^\delta \bar{\gamma}_B }{\bar{\gamma}_E}\left|  {\begin{array}{*{20}c}
    {\frac{2}{\delta},1}   \\
   { \frac{2}{\delta},0}  \\
\end{array}} \right.}  \right]\right].
\end{align}
\hrulefill
\end{figure*}
\end{proposition}
\begin{proof}
See Appendix \ref{Proof_SecrecyMetric_Discrete}.
\end{proof}
\section{Numerical results and discussion}
For the simplicity of notations, we assume $d_{AR} = 1, d_{RB} = 0.8$, $R_2 = 1$, $R_1 = 0.1$, and $\delta = 3.7$ for both phase shifting designs. Note that $\epsilon_{l,i}$ for the random discrete phase shifting design is randomly chosen from $0$ and $\pi$.

Considering the imperfect coherent phase shifting design and random discrete phase shifting design, Figs. \ref{Fig_SOP_CoherentPhase} and \ref{Fig_SOP_DiscretePhase} plot the analytical SOP with MC simulations, respectively. As shown in Fig. \ref{Fig_SOP_CoherentPhase}, the lower bound of SOP presents a gradually tight approximation to the exact SOP as $\bar{\gamma}_B$ or $L$ increases. We observe that the SOP gain due to the increase of the RIS elements number gradually decreases.

Unlike the coherent phase shifting design, one can observe that (i) the analytical SOP $\mathcal{P}_{out}^D$ becomes tight to the simulated ones as $L$ increases due to the condition that $f_E^D(\gamma)$ and $F_E^D(\gamma)$ are valid as $L$ tends to large values; (ii) as suggested by (\ref{LowSOP_exact_Discrete}), the lower bound of SOP, $\mathcal{P}_{out}^{D,L}$, presents a similar behavior as the PNZ, $\mathcal{P}_{nz}^{D}$, shown in Fig. \ref{Fig_pnzD_vs_gammaB}, i.e., both analytical secrecy metrics do not vary as $L$ increases; and (iii) $\mathcal{P}_{out}^{D,L}$ becomes tight to the analytical SOP $\mathcal{P}_{out}^{D}$ under the condition that 1) $\bar{\gamma}_B$ tends to $\infty$; 2) $R_t$ goes to $0$; and 3) $L$ goes to large values.

In Figs. \ref{Fig_pnzC_vs_gammaB} and \ref{Fig_pnzD_vs_gammaB}, we plot the PNZ for both cases. We observe that (i) our analytical results are verified by simulations; (ii) with the increase of the RIS elements number $L$, there is no obvious PNZ performance gain obtained. Different from the aforementioned two figures, we plot the PNZ against $L$ in Fig. \ref{Fig_pnz_vs_RISnumber}, it is interesting to observe that the PNZ remains constant for as $L$ increases, in other words, the random discrete phase shifting design behaves the same as the imperfect coherent phase shifting design in terms of PNZ. As a result, one can obtain an useful insight that the random discrete phase shifting design outperforms the imperfect phase shifting design with reduced system overhead and without presenting secrecy performance loss. 
\begin{figure}[!h]
\centering{\includegraphics[width=\columnwidth]{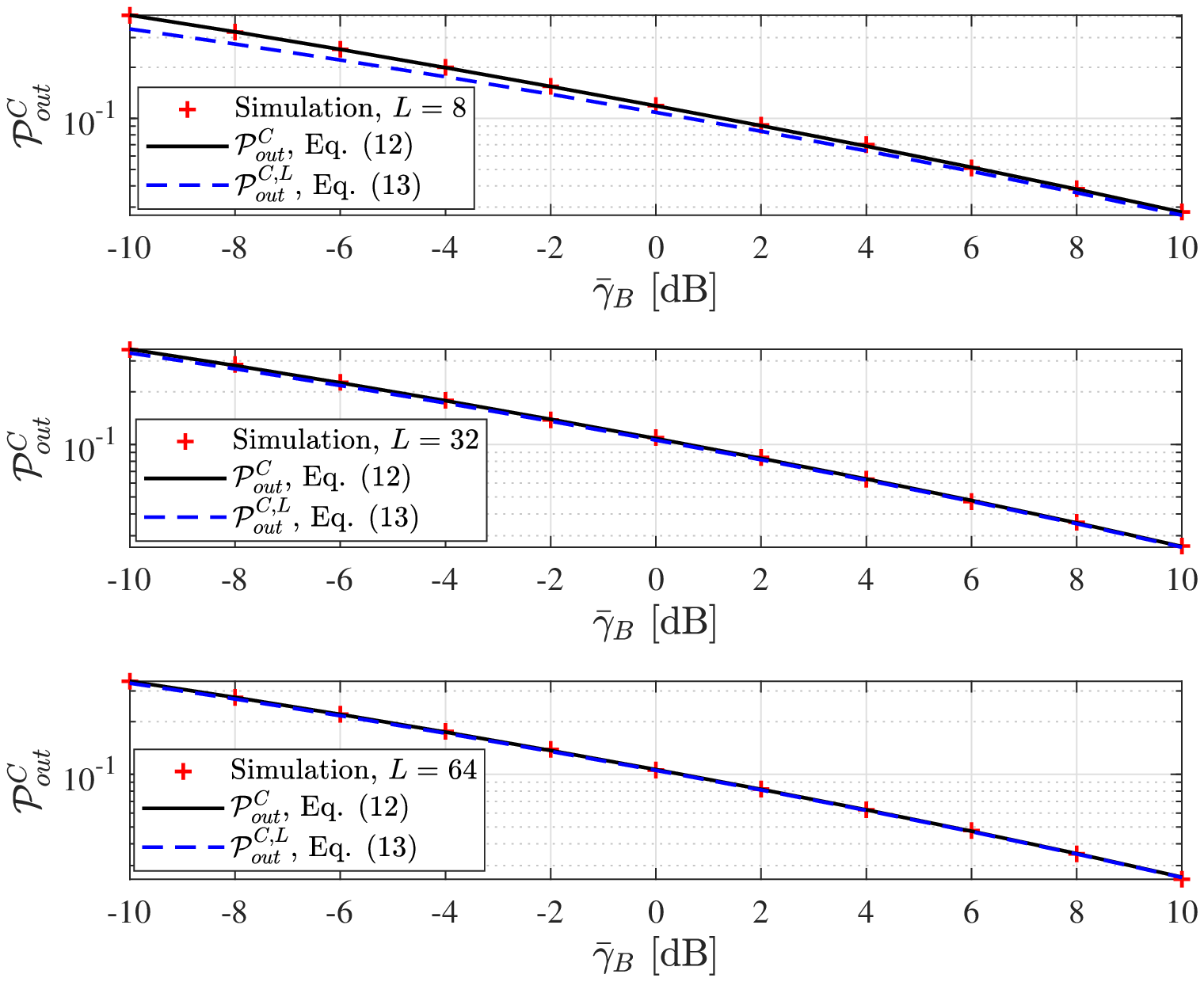}}
\caption{$\mathcal{P}_{out}$ versus $\bar{\gamma}_B$ for selected values of $L$ when $R_t = 0.8$ and $\bar{\gamma}_E= -20$ dB.}
\label{Fig_SOP_CoherentPhase}
\end{figure}
\begin{figure}[!t]
\centering{\includegraphics[width=\columnwidth]{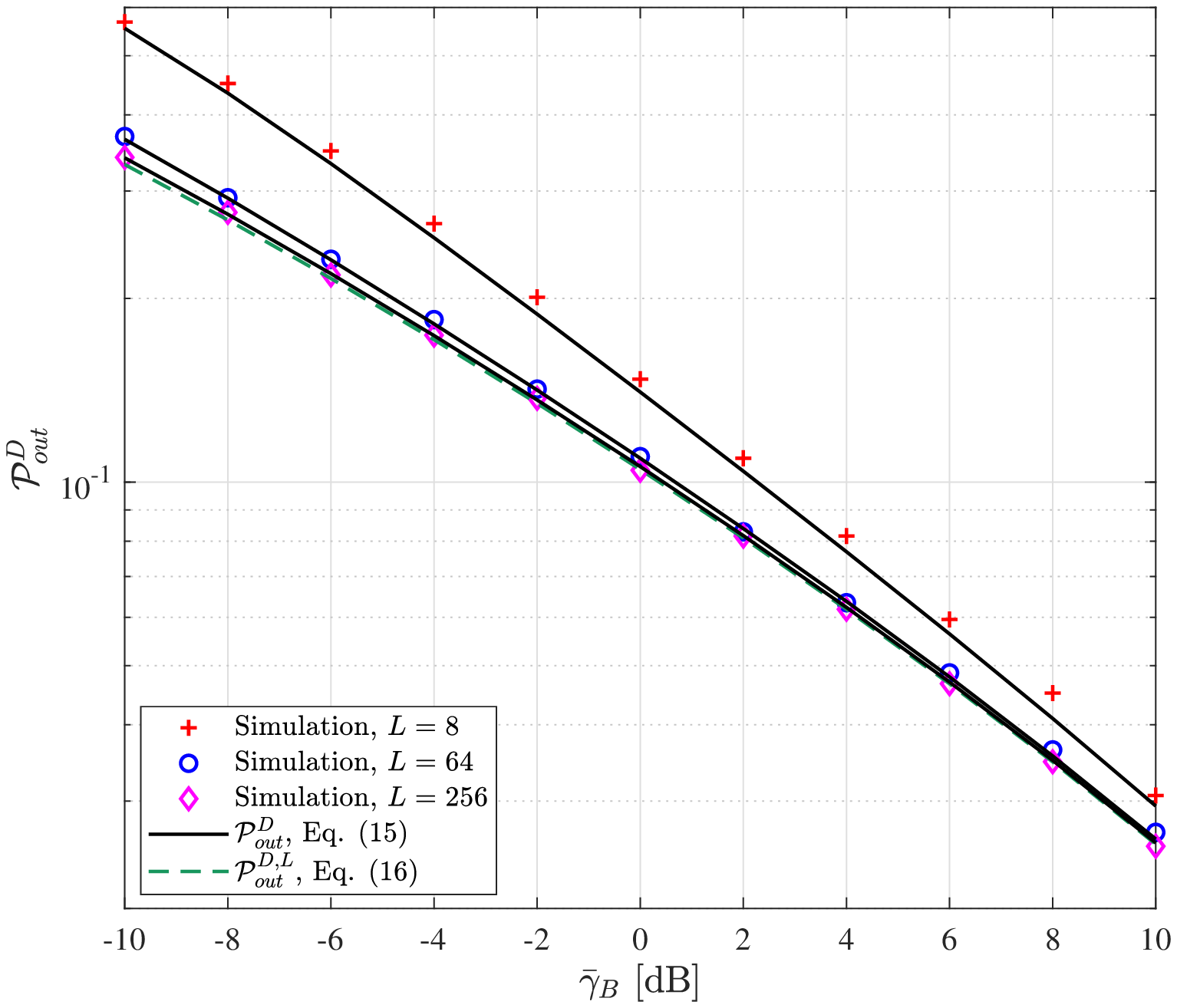}}
\caption{$\mathcal{P}_{out}$ versus $\bar{\gamma}_B$ for selected values of $L$ when $R_t = 0.8$ and $\bar{\gamma}_E= -20$ dB.}
\label{Fig_SOP_DiscretePhase}
\end{figure}
\begin{figure}[!t]
\centering{\includegraphics[width=\columnwidth]{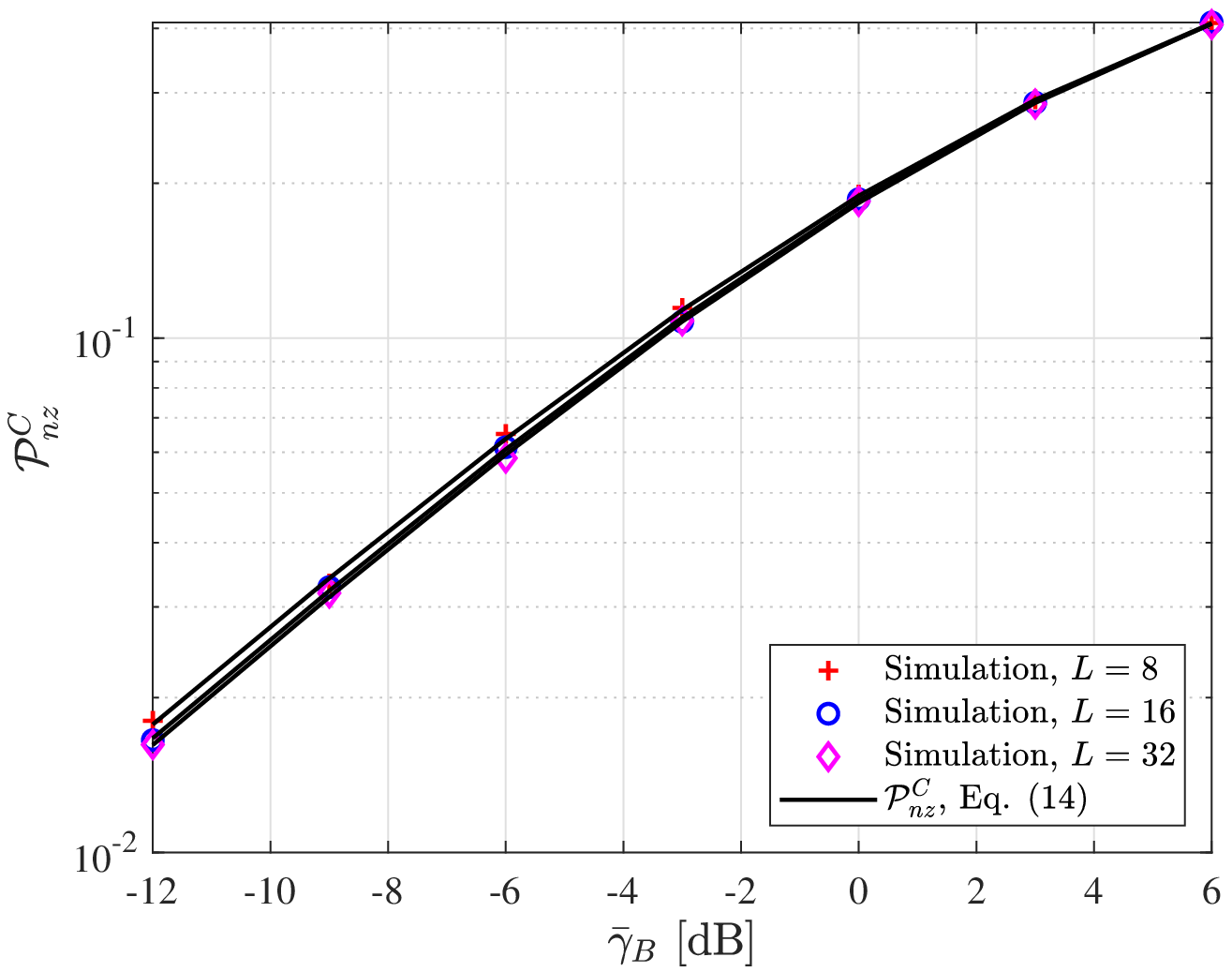}}
\caption{$\mathcal{P}_{nz}^C$ versus $\bar{\gamma}_B$ for selected values of $L$ when $\bar{\gamma}_E= 5$ dB.}
\label{Fig_pnzC_vs_gammaB}
\end{figure}
\begin{figure}[!t]
\centering{\includegraphics[width=\columnwidth]{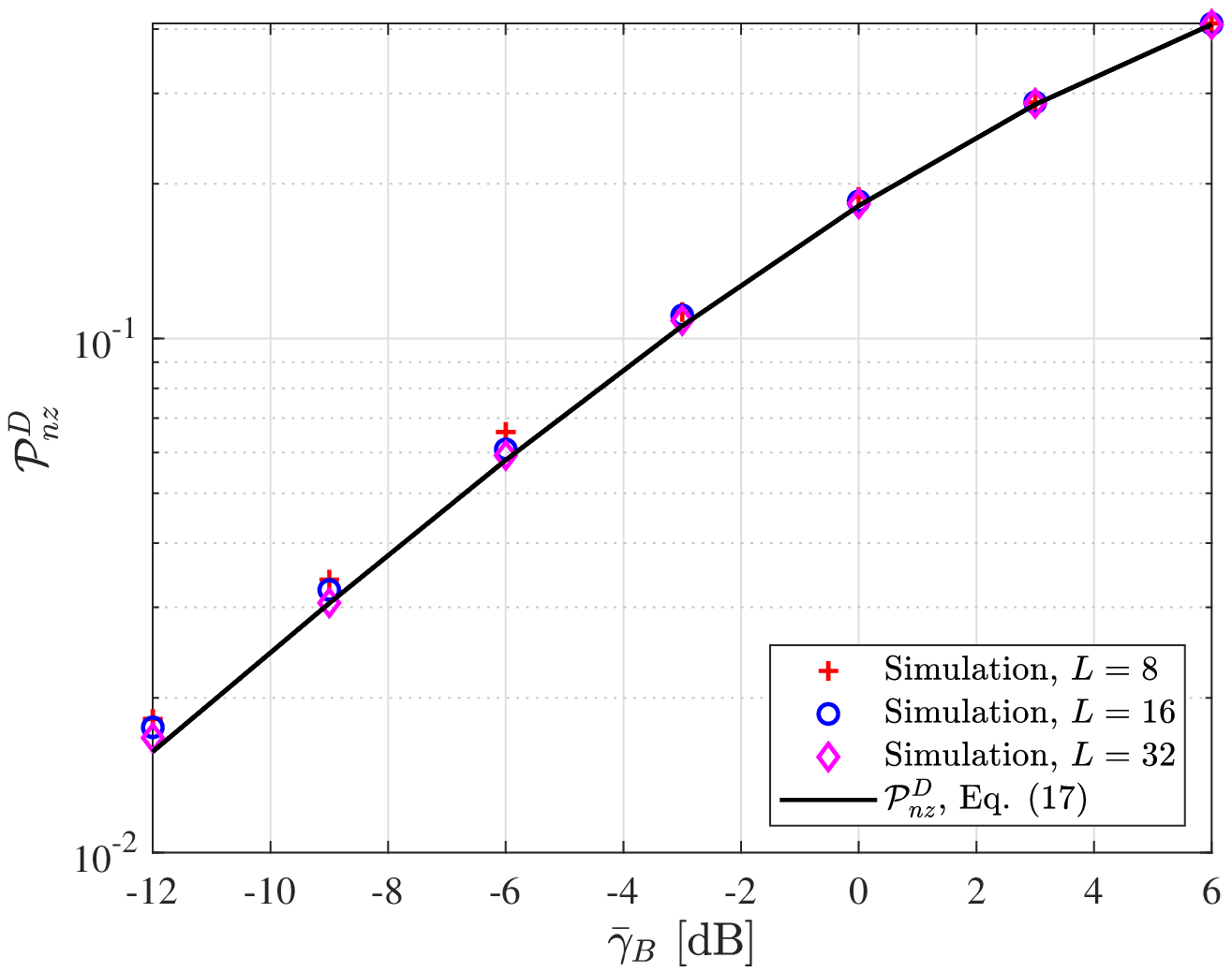}}
\caption{$\mathcal{P}_{nz}^D$ versus $\bar{\gamma}_B$ for selected values of $L$ when $\bar{\gamma}_E= 5$ dB.}
\label{Fig_pnzD_vs_gammaB}
\end{figure}
\begin{figure}[!t]
\centering{\includegraphics[width=\columnwidth]{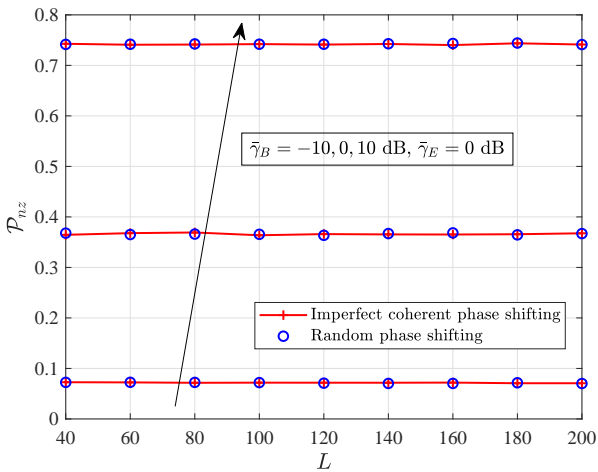}}
\caption{$\mathcal{P}_{out}$ versus $L$ for selected values of $\bar{\gamma}_B$.}
\label{Fig_pnz_vs_RISnumber}
\end{figure}
\section{Conclusion}
This paper investigated the RIS-assisted communication systems from the physical layer security perspective, where the impacts of phase shifting design and eavesdropper's location uncertainty are taken into  consideration. Secrecy metrics, including the SOP, the lower bound of SOP, and PNZ, are derived with closed-form expressions, and are further validated with the MC simulations. Numerical results demonstrate that the random discrete phase shifting design outperforms the imperfect phase shifting design with reduced complexity and same performance regarding the PNZ.  
\appendices
\section*{Acknowledgment}
This work was supported by the Luxembourg National Research Fund (FNR) under the CORE project RISOTTI C20/IS/14773976.
\section{Proof for Proposition  \ref{Proposition_GammaE_Discrete}} \label{Proof_Discrete_SNR_E}
Let $X=|\epsilon|^2$, $Y = d_{RE}^\delta$, and $\gamma_E = \bar{\gamma}_E\frac{X}{Y}$. The PDF of $Y$ is given by \cite[Eq. (28)]{9622149}
\begin{equation} \label{PDF_Y}
f_Y(y) = \frac{2y^{\frac{2}{\delta}-1  }}{\delta(R_2^2-R_1^2)}, \quad R_1^\delta \le y \le R_2^\delta.
\end{equation}
Since $X$ follows the exponential distribution with parameter $L$, and then substituting the PDF of $X$ and (\ref{PDF_Y}) into the following
\begin{equation}
f_E^D(\gamma) =  \frac{1}{\bar{\gamma}_E} \int_{R_1^\delta}^{R_2^\delta} y f_X\left( \frac{\gamma y}{\bar{\gamma}_E} \right)f_Y(y)dy,
\end{equation}
and applying \cite[Eq.(3.381.1)]{gradshteyn2014table}, the proof for $f_E^D(\gamma)$ is completed.

The CDF of $\gamma_E$ is derived 
\begin{equation}
F_E^D(\gamma) = \int_{R_1^\delta}^{R_2^\delta}F_X\left(\frac{y\gamma}{\bar{\gamma}_E} \right)f_Y(y)dy,
\end{equation}
using \cite[Eq.(3.381.1)]{gradshteyn2014table}, the proof for $F_E^D(\gamma)$ is obtained.
\section{Proof for SOP} \label{Proof_SecrecyMetric_uncertainty}
Substituting (\ref{CDF_gammaB}) and (\ref{PDF_gammaE}) into the mathematical expression of the SOP \cite[Eq. (8)]{8706707}
\begin{align} \label{SOP_Def}
\mathcal{P}_{out}^C & = \int_0^\infty F_B^C(R_s \gamma + W)f_E^C(\gamma) d\gamma,
\end{align}
next making the change of variables $\frac{R_s \gamma + W}{4\bar{\gamma}_B} = y$ and using \cite[Eq.(2.24.1.3)]{prudnikov1990integrals}, after some algebraic manipulations, the proof for the exact SOP is achieved.

Revisiting the definition of the lower bound of SOP \cite[Eq. (14)]{8706707}, and substituting (\ref{CDF_gammaB}) and (\ref{PDF_gammaE}) into it, we have 
\begin{align} \label{LowSOP_Def}
\mathcal{P}_{out}^{C,L} &= \int_0^\infty F_B^C(R_s \gamma)f_E^C(\gamma) d\gamma  \nonumber\\
& = 1- \frac{R_2^{2+\delta} I_1 -R_1^{2+\delta} I_2  }{2\delta \bar{\gamma}_E \Gamma(L)\Gamma(L)(R_2^2 - R_1^2)},
\end{align}
where 
\begin{equation}
I_1 = \int_0^\infty G_{0,2}^{2,0} \left[ { \frac{R_s \gamma}{4\bar{\gamma}_B}\left| \hspace{-1.3ex} {\begin{array}{*{20}c}
    {-}   \\
   { L,0}  \\
\end{array}} \right.} \hspace{-1.3ex} \right] G_{1,3}^{2,1} \left[ { \frac{R_2^\delta \gamma}{4\bar{\gamma}_E}\left| \hspace{-1.3ex} {\begin{array}{*{20}c}
    {-\frac{2}{\delta}}   \\
   { L-1,0,-1 - \frac{2}{\delta}}  \\
\end{array}} \right.} \hspace{-1.3ex} \right]d\gamma,
\end{equation}
\begin{equation}
I_2 = \int_0^\infty G_{0,2}^{2,0} \left[ { \frac{R_s \gamma}{4\bar{\gamma}_B}\left| \hspace{-1.3ex} {\begin{array}{*{20}c}
    {-}   \\
   { L,0}  \\
\end{array}} \right.} \hspace{-1.3ex} \right] G_{1,3}^{2,1} \left[ { \frac{R_1^\delta \gamma}{4\bar{\gamma}_E}\left| \hspace{-1.3ex} {\begin{array}{*{20}c}
    {-\frac{2}{\delta}}   \\
   { L-1,0,-1 - \frac{2}{\delta}}  \\
\end{array}} \right.} \hspace{-1.3ex} \right]d\gamma,
\end{equation}
then using \cite[Eq. (7.811.1)]{gradshteyn2014table}, the proof for $I_1$ and $I_2$ are obtained, afterwards, using some algebraic manipulations, the proof for the lower bound of SOP is completed.

Similarly, substituting (\ref{PDF_gammaB}) and (\ref{CDF_gammaE}) into the mathematical definition of the PNZ \cite[Eq. (14)]{8706707}
\begin{align}
\mathcal{P}_{nz}^C = \int_0^\infty f_B^C(\gamma )F_E^C(\gamma) d\gamma,
\end{align}
and then following the same methodology to obtain $\mathcal{P}_{out}^{C,L}$, the PNZ is achieved.
\section{Proof for SOP discrete phase shifting} \label{Proof_SecrecyMetric_Discrete}
Revisiting the mathematical expression of the SOP, and then substituting (\ref{CDF_gammaB_Discrete}) and (\ref{PDF_gammaE_Discrete}) into (\ref{SOP_Def}) yields
\begin{align}
\mathcal{P}_{out}^D & = 1 - \frac{2\left( L\bar{\gamma}_E\right)^{\frac{2}{\delta}}}{\delta(R_2^2 -R_1^2)} \int_0^\infty x^{-1-\frac{2}{\delta}}\exp\left(-\frac{R_sx +W}{L\bar{\gamma}_B} \right) \nonumber \\
& \hspace{2ex}\times \left[ \gamma\left(1+\frac{2}{\delta},\frac{LR_2^\delta}{\bar{\gamma}_E}x \right)-\gamma\left(1+\frac{2}{\delta},\frac{LR_1^\delta}{\bar{\gamma}_E}x \right) \right] dx \nonumber \\
& \mathop=^{(b)} 1 - \frac{2}{\delta(R_2^2 -R_1^2)} \left( \frac{R_s\bar{\gamma}_E}{\bar{\gamma}_B} \right)^{\frac{2}{\delta}}(\mathcal{I}_3 - \mathcal{I}_4),
\end{align}
where step $(b)$ is obtained by interchanging variables $\frac{R_sx}{L\bar{\gamma}_B} = y$, and then re-expressing the exponential function and the lower incomplete Gamma function in the manner of Meijer's $G$-function, i.e., $\exp(-x) = G_{0,1}^{1,0} \left[ { x\left|  {\begin{array}{*{20}c}
    {-}   \\
   { 0}  \\
\end{array}} \right.}  \right]$ and $\gamma(a,x) = G_{1,2}^{1,1} \left[ { x\left|  {\begin{array}{*{20}c}
    {1}   \\
   { a,0}  \\
\end{array}} \right.}  \right]$, we have \\  
\begin{align}
    \mathcal{I}_3 =&  \int_0^\infty y^{-1-\frac{2}{\delta}} G_{0,1}^{1,0} \left[ { y + \frac{W}{L\bar{\gamma}_B}\left|\hspace{-1.2ex}  {\begin{array}{*{20}c}
    {-}   \\
   { 0}  \\
\end{array}} \right.} \hspace{-1.2ex} \right] \nonumber \\
& \hspace{5ex}\times G_{1,2}^{1,1} \left[ { \frac{\bar{\gamma}_BR_2^\delta}{R_s\bar{\gamma}_E}y\left| \hspace{-1.2ex} {\begin{array}{*{20}c}
    {1}   \\
   { 1 + \frac{2}{\delta},0}  \\
\end{array}} \right.} \hspace{-1.2ex} \right]dy,
\end{align}
\begin{align}
\mathcal{I}_4 =& \int_0^\infty y^{-1-\frac{2}{\delta}} G_{0,1}^{1,0} \left[ { y + \frac{W}{L\bar{\gamma}_B}\left| \hspace{-1.2ex} {\begin{array}{*{20}c}
    {-}   \\
   { 0}  \\
\end{array}} \right.} \hspace{-1.2ex} \right] \nonumber \\ 
& \hspace{5ex} \times G_{1,2}^{1,1} \left[ { \frac{\bar{\gamma}_BR_1^\delta}{R_s\bar{\gamma}_E}y\left|\hspace{-1.2ex}  {\begin{array}{*{20}c}
    {1}   \\
   { 1 + \frac{2}{\delta},0}  \\
\end{array}} \right.} \hspace{-1.2ex} \right]dy.
\end{align}
$\mathcal{I}_3$ is developed using \cite[Eq.(2.24.1.3)]{prudnikov1990integrals} and the property of Meijer's $G$-function \cite[Eq. (9.31.1)]{gradshteyn2014table}. $\mathcal{I}_4$ can be similarly obtained. Finally, making some manipulations, the proof for $\mathcal{P}_{out}^D$ is completed.

Referring to the definition of the lower bound of SOP, we have
\begin{align}
\mathcal{P}_{out}^{D,L} & = \int_0^\infty F_B^D(R_s \gamma)f_E^D(\gamma) d\gamma  \nonumber\\
&=1-\frac{2L}{\bar{\gamma}_E \delta (R_2^2 -R_1^2)} \int_0^\infty \left(\frac{L\bar{\gamma}_E}{\gamma} \right)^{1+\frac{2}{\delta}}\exp\left(-\frac{R_s\gamma}{L\bar{\gamma}_B} \right) \nonumber\\
& \hspace{2ex} \times \left[ \gamma \left(1+\frac{2}{\delta},\frac{LR_2^\delta}{\bar{\gamma}_E}\gamma \right) - \gamma \left(1+\frac{2}{\delta},\frac{LR_1^\delta}{\bar{\gamma}_E}\gamma \right) \right] d\gamma, 
\end{align}
next re-expressing the exponential function and the lower incomplete Gamma function in the manner of Meijer's $G$-function, the proof for $\mathcal{P}_{out}^{D,L}$ is obtained with the assistance of \cite[Eq. (7.813.1)]{gradshteyn2014table}.

Following the steps to derive $\mathcal{P}_{out}^{D,L}$, the proof for $\mathcal{P}_{nz}^D$ can be similarly developed. 

\ifCLASSOPTIONcaptionsoff
  \newpage
\fi

\balance

\vspace{-0.2cm}
\bibliographystyle{IEEEtran}
\bibliography{EUSIPCO2021}
\end{document}